\documentclass[letterpaper, 10 pt, journal, twoside]{IEEEtran}
\usepackage[utf8]{inputenc}
\usepackage{amsmath} 
\usepackage{amssymb}
\usepackage{color}
\usepackage[algo2e,vlined,linesnumbered]{algorithm2e}
\usepackage{graphicx}

\makeatletter

\let\proof\@undefined
\let\endproof\@undefined
\makeatother
\usepackage{amsthm}
\newcommand{\R}{\mathbb{R}}
\newcommand{\N}{\mathbb{N}}

\newtheorem{assum}{Assumption}
\newtheorem{problem}[assum]{Problem}
\newtheorem{remark}[assum]{Remark}
\newtheorem{theorem}[assum]{Theorem}
\newtheorem{coro}[assum]{Corollary}

\newtheorem{lemma}[assum]{Lemma}

\graphicspath{{Images/}}

\title{Sampled-data reachability analysis\\ using sensitivity and mixed-monotonicity}

\author{Pierre-Jean Meyer$^{1}$, Samuel Coogan$^{2}$ and Murat Arcak$^{1}$% <-this % stops a space
\thanks{Funded in part by the National Science Foundation grant CNS-1446145.}% <-this % stops a space
\thanks{$^{1}$P.-J.\ Meyer and M.\ Arcak are with the Department of Electrical Engineering and Computer Sciences, University of California, Berkeley, USA.
{\tt\small $\{$pjmeyer,arcak$\}$@berkeley.edu}}%
\thanks{$^{2}$S. Coogan is with the School of Electrical and Computer Engineering and the School of Civil and Environmental Engineering at the Georgia Institute of Technology, Atlanta, USA.
{\tt\small sam.coogan@gatech.edu}}%
}

\begin{document}
\maketitle
\thispagestyle{empty}
\pagestyle{empty}

%%%%%%%%%%%%%%%%%%%%%%%%%%%%%%%%%%%%%%%%%%%
%%%%%%%%%%%%%%%%%%%%%%%%%%%%%%%%%%%%%%%%%%%

\begin{abstract}
This paper over-approximates the reachable sets of a continuous-time uncertain system using the sensitivity of its trajectories with respect to initial conditions and uncertain parameters.
We first prove the equivalence between an existing over-approximation result based on the sign-stability of the sensitivity matrices and a discrete-time approach relying on a mixed-monotonicity property.
We then present a new over-approximation result which scales at worst linearly with the state dimension and is applicable to any continuous-time system with bounded sensitivity.
Finally, we provide a simulation-based approach to estimate these bounds through sampling and falsification.
The results are illustrated with numerical examples on traffic networks and satellite orbits.

% at most 200 words and 1800 characters
\end{abstract}

%%%%%%%%%%%%%%%%%%%%%%%%%%%%%%%%%%%%%%%%%%%
%%%%%%%%%%%%%%%%%%%%%%%%%%%%%%%%%%%%%%%%%%%

\begin{IEEEkeywords}
Numerical algorithms, Uncertain systems.
\end{IEEEkeywords}

%%%%%%%%%%%%%%%%%%%%%%%%%%%%%%%%%%%%%%%%%%%
%%%%%%%%%%%%%%%%%%%%%%%%%%%%%%%%%%%%%%%%%%%
\section{Introduction}
\label{sec intro}
\IEEEPARstart{R}{eachability} analysis deals with the problem of computing the set of all possible successors of a system given its sets of initial conditions and admissible disturbance and uncertainty values (see e.g.~\cite{blanchini2008set,rakovic2006reachability}).
Since exact computation of the reachable set is rarely possible, 
we instead evaluate an over-approximation to guarantee that the obtained set contains all possible successors of the system.
Various methods and representations exist for these over-approximations, including ellipsoids~\cite{kurzhanskiy2007ellipsoidal}, polytopes~\cite{chutinan2003computational}, zonotopes~\cite{althoff2010computing}, level-sets~\cite{mitchell2000level} and unions of intervals~\cite{jaulin2001applied}.
Their main focus is on obtaining over-approximations as close as possible to the actual reachable set, which can then be used for safety verification to ensure that a bad set is never crossed (see e.g.~\cite{frehse2005phaver}).

Alternatively, methods using a single interval such as~\cite{ramdani2009hybrid} focus less on the quality of the over-approximations and more on the simplicity of implementation, including features such as low memory usage (only two states) and low complexity of the reachability analysis (at best constant for monotone systems~\cite{angeli_monotone}, at worst linear in the state dimension~\cite{xue2017cdc}).
These properties are particularly important in the context of abstraction-based control synthesis (see, e.g.,~\cite{coogan2015efficient}) where a large number of over-approximations have to be computed, stored and intersected with other intervals.

This paper focuses on the computation of interval over-approximations of reachable sets for a continuous-time uncertain system.
As opposed to monotonicity-based approaches relying on the sign of the Jacobian matrices~\cite{ramdani2009hybrid,coogan2015efficient}, the proposed approach uses the sensitivity matrices (partial derivatives of the system trajectories with respect to the initial state or uncertain parameters).
Such an approach was introduced in~\cite{xue2017cdc} for the case of systems whose sensitivity matrix is sign-stable over the set of initial states.

This paper presents three main contributions.
1) In Section~\ref{sec sign stable}, we prove the equivalence between the sign-stable sensitivity approach in~\cite{xue2017cdc} for continuous-time systems and the one based on mixed-monotonicity for discrete-time systems in~\cite{coogan2015efficient}.
2) We next propose in Section~\ref{sec bounded} a generalized sensitivity-based reachability analysis applicable to any continuous-time system whose sensitivity matrices are bounded. This generalization is motivated by the one introduced in~\cite{yang2017note} for continuous-time mixed-monotone systems.
3) Since the proposed approach is based on the system trajectories and sensitivity, which are unknown for most continuous-time systems, we lastly present a simulation-based method to estimate the sensitivity bounds using sampling and falsification in Section~\ref{sec bounds}.
Section~\ref{sec simu} then illustrates these results through an example of traffic flow on a road network and an example of a satellite orbit.

%%%%%%%%%%%%%%%%%%%%%%%%%%%%%%%%%%%%%%%%%%%
%%%%%%%%%%%%%%%%%%%%%%%%%%%%%%%%%%%%%%%%%%%
% \section{Preliminaries}
\section{Problem formulation}
\label{sec prelim}
Let $\R$ be the set of reals and $\mathcal{I}\subseteq 2^\R$ the set of closed real intervals, i.e.,\ for all $X\in\mathcal{I}$, there exist $\underline{x},\overline{x}\in\R$ such that $X=[\underline{x},\overline{x}]=\{x\in\R~|~\underline{x}\leq x\leq \overline{x}\}\subseteq\R$.
$\mathcal{I}^n$ and $\mathcal{I}^{n\times q}$ then represent the sets of interval vectors in $\R^n$ and interval matrices in $\R^{n\times q}$, respectively.

%%%%%%%%%%%%%%%%%%%%%%%%%%%%%%%%%%%%%%%%%%%
% \subsection{System definition}
% \label{sub prelim system}
We consider a continuous-time, time-varying system
\begin{equation}
\label{eq system}
\dot x=f(t,x,p),
\end{equation}
with state $x\in\R^n$, uncertain parameter $p\in\R^q$ and continuously differentiable vector field $f:\R\times\R^n\times\R^q\rightarrow\R^n$.
We denote as $\Phi(t;t_0,x_0,p)\in\R^n$ the state reached by (\ref{eq system}) at time $t\geq t_0$ from initial state $x_0$ with parameter $p$.
The variable $p$ can also represent control or disturbance parameters that remain constant over the considered time interval $[t_0,t]$.
Given sets $X_0\subseteq\R^n$ and $P\subseteq\R^q$ of initial states and parameters, respectively, the reachable set of (\ref{eq system}) at time $t\geq t_0$ is denoted as
\begin{equation}
\label{eq reachable set}
R(t;t_0,X_0,P)=\{\Phi(t;t_0,x_0,p)~|~x_0\in X_0,~p\in P\}.
\end{equation}
  
The sensitivity of the trajectories of (\ref{eq system}) with respect to the initial conditions and parameters are defined as
\begin{gather}
\label{eq sensitivity state}
s^x(t;t_0,x_0,p)=\frac{\partial\Phi(t;t_0,x_0,p)}{\partial x_0}\in\R^{n\times n},\\
\label{eq sensitivity parameter}
s^p(t;t_0,x_0,p)=\frac{\partial\Phi(t;t_0,x_0,p)}{\partial p}\in\R^{n\times q}.
\end{gather}
% \begin{equation}
% \label{eq sensitivity state}
% s^x(t;t_0,x_0,p)=\frac{\partial\Phi(t;t_0,x_0,p)}{\partial x_0}\in\R^{n\times n},
% \end{equation}
% \begin{equation}
% \label{eq sensitivity parameter}
% s^p(t;t_0,x_0,p)=\frac{\partial\Phi(t;t_0,x_0,p)}{\partial p}\in\R^{n\times q}.
% \end{equation}
The sensitivities defined in (\ref{eq sensitivity state}) and (\ref{eq sensitivity parameter}) thus represent the differential influence of the initial conditions and parameters, respectively, on the successor of (\ref{eq system}) at time $t$.

%%%%%%%%%%%%%%%%%%%%%%%%%%%%%%%%%%%%%%%%%%%
%%%%%%%%%%%%%%%%%%%%%%%%%%%%%%%%%%%%%%%%%%%
% \subsection{Problem formulation}
% \label{sub prelim problem}
Our objective is to compute an over-approximation of the reachable set (\ref{eq reachable set}) at time $T\geq t_0$ for intervals of initial conditions $X_0\subseteq\R^n$ and of possible parameters $P\subseteq\R^q$.
\begin{problem}
\label{pb oa}
Given times $t_0\in\R$ and $T\geq t_0$ and intervals $X_0\in\mathcal{I}^n$ and $P\in\mathcal{I}^q$, find a set $\bar R(T;t_0,X_0,P)\subseteq\R^n$ such that $R(T;t_0,X_0,P)\subseteq\bar R(T;t_0,X_0,P)$.
\end{problem}

%%%%%%%%%%%%%%%%%%%%%%%%%%%%%%%%%%%%%%%%%%%
%%%%%%%%%%%%%%%%%%%%%%%%%%%%%%%%%%%%%%%%%%%
\section{Reachability with sign-stable sensitivity}
\label{sec sign stable}
In this section, we review the over-approximation approach presented in~\cite{xue2017cdc} with the aim of connecting it in Section~\ref{sub stable mixed mono} to 
discrete-time mixed-monotonicity from~\cite{coogan2015efficient}.

%%%%%%%%%%%%%%%%%%%%%%%%%%%%%%%%%%%%%%%%%%%
\subsection{Sensitivity-based reachability analysis}
\label{sub stable sensitivity}
Reference~\cite{xue2017cdc} provides a method to obtain an interval over-approximation of the reachable set for an autonomous system $\dot x=f(x)$ whose sensitivity matrix $s^x(T,x_0)$ at time $T\geq0$ is \emph{sign-stable} over the set of initial states $X_0$.
For the purpose of the comparison in Section~\ref{sub stable mixed mono}, these results are reviewed here in a more general framework where the system (\ref{eq system}) depends on both time and an uncertain parameter.

We assume that the sensitivity matrices defined in (\ref{eq sensitivity state}) and (\ref{eq sensitivity parameter}) at time $T$ are sign-stable over the sets $X_0$ and $P$, i.e.\ their entries do not change sign when the initial state and parameter vary in $X_0$ and $P$. 
This is formalized as follows.
\begin{assum}
\label{assum stable}
For all $x_0,\tilde x_0\in X_0$, $p,\tilde p\in P$, $i,j\in\{1,\dots,n\}$, $k\in\{1,\dots,q\}$, we have 
$$
\begin{cases}
s^x_{ij}(T;t_0,x_0,p)s^x_{ij}(T;t_0,\tilde x_0,\tilde p)\geq0,\\
s^p_{ik}(T;t_0,x_0,p)s^p_{ik}(T;t_0,\tilde x_0,\tilde p)\geq0.
\end{cases}
$$
\end{assum}

As $X_0$ and $P$ are intervals, let $\underline{x_0},\overline{x_0}\in\R^n$ and $\underline{p},\overline{p}\in\R^q$ be such that $X_0=[\underline{x_0},\overline{x_0}]$ and $P=[\underline{p},\overline{p}]$.
For each $i\in\{1,\dots,n\}$, define states $\underline{\xi}^i,\overline{\xi}^i\in X_0$ (with, e.g.,  $\underline{\xi}^i=[\underline{\xi}^i_1;\dots;\underline{\xi}^i_n]$) and parameters $\underline{\pi}^i,\overline{\pi}^i\in P$ as diagonally opposite vertices of $X_0$ and $P$, respectively, where for each $j\in\{1,\dots,n\}$ (resp.\ $k\in\{1,\dots,q\}$), their entries $\underline{\xi}^i_j$, $\overline{\xi}^i_j$ (resp.\ $\underline{\pi}^i_k$, $\overline{\pi}^i_k$) are allocated to $\underline{x_0}_j$ or $\overline{x_0}_j$ (resp.\ $\underline{p}_k$ or $\overline{p}_k$) based on the sign of the sensitivity $s^x_{ij}$ (resp.\ $s^p_{ik}$):
\begin{equation}
\label{eq stable state-param}
\begin{aligned}
(\underline{\xi}^i_j,\overline{\xi}^i_j)=
\begin{cases}
(\underline{x_0}_j,\overline{x_0}_j)&\text{ if }s^x_{ij}(T;t_0,x_0,p)\geq0,\\
(\overline{x_0}_j,\underline{x_0}_j)&\text{ if }s^x_{ij}(T;t_0,x_0,p)<0,\\
\end{cases}\\
(\underline{\pi}^i_k,\overline{\pi}^i_k)=
\begin{cases}
(\underline{p}_k,\overline{p}_k)&\text{ if }s^p_{ik}(T;t_0,x_0,p)\geq0,\\
(\overline{p}_k,\underline{p}_k)&\text{ if }s^p_{ik}(T;t_0,x_0,p)<0.
\end{cases}
\end{aligned}
\end{equation}

From (\ref{eq stable state-param}) and the sensitivity definitions in (\ref{eq sensitivity state}) and (\ref{eq sensitivity parameter}), the successor $\Phi(T;t_0,\underline\xi^i,\underline{\pi}^i)$ (resp.\ $\Phi(T;t_0,\overline{\xi}^i,\overline{\pi}^i)$) is guaranteed to define the lower bound (resp.\ upper bound) of the reachable set $R(T;t_0,X_0,P)$ on dimension $i$.
\begin{lemma}[\cite{xue2017cdc}]
\label{lemma stable sensi}
Under Assumption~\ref{assum stable}, an over-approximation $\bar R(T;t_0,X_0,P)\in\mathcal{I}^n$ of the reachable set is given in each dimension $i\in\{1,\dots,n\}$ by
\begin{equation*}
  \bar R_i(T;t_0,X_0,P)=[\Phi_i(T;t_0,\underline{\xi}^i,\underline{\pi}^i),\Phi_i(T;t_0,\overline{\xi}^i,\overline{\pi}^i)]\in\mathcal{I}.
\end{equation*}
\end{lemma}

\begin{remark}
\label{rmk mono}
Lemma~\ref{lemma stable sensi} requires computing the full successor $\Phi$ of (\ref{eq system}) for each pair $(\underline{\xi}^i,\underline{\pi}^i),(\overline{\xi}^i,\overline{\pi}^i)\in X_0\times P$ even though only the $i^{th}$ entry $\Phi_i$ is used in the over-approximation.
More than one entry of a successor $\Phi$ may be used when there exists less than $2n$ distinct pairs.
The computational burden can thus go from at most $2n$ successors when all the above pairs are distinct, to at least $2$ successors when there exist $\tilde x,\hat x\in X_0$ and $\tilde p,\hat p\in P$ such that 
$(\underline{\xi}^i,\overline{\xi}^i,\underline{\pi}^i,\overline{\pi}^i)=(\tilde x,\hat x,\tilde p,\hat p)$ or $(\underline{\xi}^i,\overline{\xi}^i,\underline{\pi}^i,\overline{\pi}^i)=(\hat x,\tilde x,\hat p,\tilde p)$ for all $i\in\{1,\dots,n\}$.
The latter case corresponds to continuous-time monotonicity of (\ref{eq system}) with respect to orthants, as described in~\cite{angeli_monotone}.
\end{remark}

Note also that the over-approximation obtained in Lemma~\ref{lemma stable sensi} is \emph{tight} in the sense that $\bar R(T;t_0,X_0,P)$ is the smallest interval in $\mathcal{I}^n$ containing the reachable set.
\begin{coro}
\label{coro tightness}
For all $X\in\mathcal{I}^n$, if $R(T;t_0,X_0,P)\subseteq X$ then $\bar R(T;t_0,X_0,P)\subseteq X$.
\end{coro}
\begin{proof}
% For all $i\in\{1,\dots,n\}$ we have $\underline{\xi}^i,\overline{\xi}^i\in X_0$ and $\underline{\pi}^i,\overline{\pi}^i\in P$ from (\ref{eq stable state-param}), thus leading to $\Phi(T;t_0,\underline{\xi}^i,\underline{\pi}^i),\Phi(T;t_0,\overline{\xi}^i,\overline{\pi}^i)\in R(T;t_0,X_0,P)$.
For all $i\in\{1,\dots,n\}$ we have $\underline{\xi}^i,\overline{\xi}^i\in X_0$ and $\underline{\pi}^i,\overline{\pi}^i\in P$ from (\ref{eq stable state-param}), thus leading to $\Phi(T;t_0,\underline{\xi}^i,\underline{\pi}^i)\in R(T;t_0,X_0,P)$ and $\Phi(T;t_0,\overline{\xi}^i,\overline{\pi}^i)\in R(T;t_0,X_0,P)$.
Since components $i$ of these reachable states define $\bar R_i(T;t_0,X_0,P)$ in Lemma~\ref{lemma stable sensi}, any interval $X\in\mathcal{I}^n$ strictly contained in $\bar R(T;t_0,X_0,P)$ cannot contain the whole reachable set $ R(T;t_0,X_0,P)$.
\end{proof}

%%%%%%%%%%%%%%%%%%%%%%%%%%%%%%%%%%%%%%%%%%%
\subsection{Comparison with discrete-time mixed-monotonicity}
\label{sub stable mixed mono}
In this section, we show that the approach described in Section~\ref{sub stable sensitivity} for the over-approximation of \emph{continuous-time} systems with sign-stable sensitivity is equivalent to the method presented in~\cite{coogan2015efficient} for \emph{discrete-time} systems satisfying a mixed-monotonicity property.
A mixed-monotone system $x^+=F(t,x,p)$ is one that is decomposable into its increasing and decreasing components and can be characterized as having sign-stable Jacobian matrices $\partial F/\partial x$ and $\partial F/\partial p$.
The reader is referred to~\cite{coogan2015efficient} for the formal definition of a mixed-monotone system and its over-approximation method.

\begin{theorem}
\label{th comparison}
Under Assumption~\ref{assum stable} and given the discrete-time system $x^+=F(t,x,p)$ with $F(t,x,p)=\Phi(T;t,x,p)$, the over-approximations of $R(T;t_0,X_0,P)$ in Lemma~\ref{lemma stable sensi} and of $F(t_0,X_0,P)$ in~\cite{coogan2015efficient} are equivalent.
\end{theorem}
\begin{proof}
The sign-stability in Assumption~\ref{assum stable} implies that $x^+=F(t,x,p)$ is mixed-monotone as in~\cite{coogan2015efficient}.
The equivalence then follows from the facts that both methods result in a tight interval over-approximation of the reachable set $F(t_0,X_0,P)=R(T;t_0,X_0,P)$ (\cite[Proposition 2]{coogan2015efficient} and Corollary~\ref{coro tightness}) and such tight interval is uniquely defined.
\end{proof}

%%%%%%%%%%%%%%%%%%%%%%%%%%%%%%%%%%%%%%%%%%%
%%%%%%%%%%%%%%%%%%%%%%%%%%%%%%%%%%%%%%%%%%%
\section{Reachability with bounded sensitivity}
\label{sec bounded}
We now extend the over-approximation method described in (\ref{eq stable state-param}) and Lemma~\ref{lemma stable sensi} after relaxing Assumption~\ref{assum stable}.
The new assumption (formalized below) is very mild as it now only requires each entry of the sensitivity matrices at time $T$ to lie in a bounded interval when the initial state and parameter vary in $X_0$ and $P$.
Unlike Assumption~\ref{assum stable}, each of these intervals is allowed to contain the value $0$ in its interior.
This modification is motivated by an extension of the definition of mixed-monotonicity for \emph{continuous-time} systems in~\cite{yang2017note}.
\begin{assum}
\label{assum bounded}
For all $i,j\in\{1,\dots,n\}$, $k\in\{1,\dots,q\}$, there exist $\underline{s^x_{ij}}, \overline{s^x_{ij}}, \underline{s^p_{ik}}, \overline{s^p_{ik}}\in\R$ such that for all $x_0\in X_0$, $p\in P$ we have $s^x_{ij}(T;t_0,x_0,p)\in[\underline{s^x_{ij}},\overline{s^x_{ij}}]$ and $s^p_{ik}(T;t_0,x_0,p)\in[\underline{s^p_{ik}},\overline{s^p_{ik}}]$.
\end{assum}

Since $X_0$ and $P$ are bounded sets, Assumption~\ref{assum bounded} is naturally satisfied by any system whose trajectory function $\Phi$ is continuously differentiable in its initial state and parameter.

Denoting the center of $[\underline{s^x_{ij}},\overline{s^x_{ij}}]$ and $[\underline{s^p_{ik}},\overline{s^p_{ik}}]$ as $s^{x*}_{ij}$ and $s^{p*}_{ik}$, respectively, we update the definition of the states $\underline{\xi}^i,\overline{\xi}^i\in X_0$ and parameters $\underline{\pi}^i,\overline{\pi}^i\in P$ in (\ref{eq stable state-param}) by replacing the right-hand side conditions on the sign of the sensitivity by the same conditions on the center of the sensitivity bounds:
\begin{equation}
\label{eq bounded state-param}
\begin{aligned}
(\underline{\xi}^i_j,\overline{\xi}^i_j)=
\begin{cases}
(\underline{x_0}_j,\overline{x_0}_j)&\text{ if }s^{x*}_{ij}\geq0,\\
(\overline{x_0}_j,\underline{x_0}_j)&\text{ if }s^{x*}_{ij}<0,\\
\end{cases}\\
(\underline{\pi}^i_k,\overline{\pi}^i_k)=
\begin{cases}
(\underline{p}_k,\overline{p}_k)&\text{ if }s^{p*}_{ik}\geq0,\\
(\overline{p}_k,\underline{p}_k)&\text{ if }s^{p*}_{ik}<0.
\end{cases}
\end{aligned}
\end{equation}
Note that the condition $s^{x*}_{ij}\geq0$ in the first line of (\ref{eq bounded state-param}) covers both cases where the whole interval $[\underline{s^x_{ij}},\overline{s^x_{ij}}]$ is positive (as in (\ref{eq stable state-param})) and where it is \emph{mostly positive} ($\underline{s^p_{ik}} \leq 0 \leq s^{x*}_{ij} \leq \overline{s^p_{ik}}$).

To account for the deviations from the sign-stable cases of (\ref{eq stable state-param}) that may arise through the mostly positive and mostly negative cases in (\ref{eq bounded state-param}), we introduce two row vectors $c^i=[c^i_1,\dots,c^i_n]\in\R^n$ and $d^i=[d^i_1,\dots,d^i_q]\in\R^q$ for each $i\in\{1,\dots,n\}$ defined by, for all $j\in\{1,\dots,n\}$, $k\in\{1,\dots,q\}$,
\begin{equation}
\label{eq bounded compensation}
\begin{aligned}
c^i_j=
\begin{cases}
\min(0,\underline{s^x_{ij}})&\text{ if }s^{x*}_{ij}\geq0,\\
\max(0,\overline{s^x_{ij}})&\text{ if }s^{x*}_{ij}<0,\\
\end{cases}\\
d^i_k=
\begin{cases}
\min(0,\underline{s^p_{ik}})&\text{ if }s^{p*}_{ik}\geq0,\\
\max(0,\overline{s^p_{ik}})&\text{ if }s^{p*}_{ik}<0.
\end{cases}
\end{aligned}
\end{equation}
Equation (\ref{eq bounded compensation}) means that $c^i_j=0$ in the sign-stable cases, $c^i_j=\underline{s^x_{ij}}\leq0$ in the \emph{mostly positive case} and $c^i_j=\overline{s^x_{ij}}\geq0$ in the \emph{mostly negative case}.

Without the sign-stability from Assumption~\ref{assum stable}, the successors $\Phi_i(T;t_0,\underline{\xi}^i,\underline{\pi}^i)$ and $\Phi_i(T;t_0,\overline{\xi}^i,\overline{\pi}^i)$ are not guaranteed to over-approximate dimension $i$ of the reachable set.
To compute an interval that is guaranteed to over-approximate the reachable set, the generalization of Lemma~\ref{lemma stable sensi} thus requires the addition of compensation terms as in the result below, where $\underline{\xi}^i,\overline{\xi}^i\in \R^n$ and $\underline{\pi}^i,\overline{\pi}^i\in\R^q$ are column vectors and $c^i\in\R^n$ and $d^i\in\R^q$ are row vectors.
\begin{theorem}
\label{th bounded sensi}
Under Assumption~\ref{assum bounded}, an over-approximation $\bar R(T;t_0,X_0,P)\in\mathcal{I}^n$ is given in each dimension $i\in\{1,\dots,n\}$ by:
\begin{multline*}
\bar R_i(T;t_0,X_0,P)=\\
[\Phi_i(T;t_0,\underline{\xi}^i,\underline{\pi}^i)-c^i(\underline{\xi}^i-\overline{\xi}^i)-d^i(\underline{\pi}^i-\overline{\pi}^i),\\
\Phi_i(T;t_0,\overline{\xi}^i,\overline{\pi}^i)+c^i(\underline{\xi}^i-\overline{\xi}^i)+d^i(\underline{\pi}^i-\overline{\pi}^i)].
\end{multline*}
\end{theorem}
\begin{proof}
Consider an auxiliary system whose trajectories $\hat\Phi$ are such that for all $x_0\in X_0$, $p\in P$ and $i\in\{1,\dots,n\}$ we have $\hat\Phi_i(T;t_0,x_0,p)=\Phi_i(T;t_0,x_0,p)-c^ix_0-d^ip$.
Then, from the sensitivity bounds in Assumption~\ref{assum bounded} and the definition of $c^i$ and $d^i$ in (\ref{eq bounded compensation}), the sensitivities $\hat s^x(T;t_0,x_0,p)$ and $\hat s^p(T;t_0,x_0,p)$ of this auxiliary system are sign-stable over the sets $X_0$ and $P$, i.e.\ for all $x_0\in X_0$ and $p\in P$, $\hat s^x_{ij}(T;t_0,x_0,p)=s^x_{ij}(T;t_0,x_0,p)-c^i_j\geq0$ (resp.\ $\leq0$) if $s^{x*}_{ij}\geq0$ (resp.\ $\leq0$), with similar results for $\hat s^p$.
Since $\hat s^x_{ij}(T;t_0,x_0,p)$ and $s^{x*}_{ij}$ have the same sign (and similarly for $\hat s^p_{ij}(T;t_0,x_0,p)$ and $s^{p*}_{ij}$), this also guarantees that the states $\underline{\xi}^i,\overline{\xi}^i\in X_0$ and parameters $\underline{\pi}^i,\overline{\pi}^i\in P$ obtained in (\ref{eq bounded state-param}) are the same as their hatted counterparts that would be obtained in (\ref{eq stable state-param}) for the auxiliary system.
Applying Lemma~\ref{lemma stable sensi} to $\hat\Phi$ implies that for all $i\in\{1,\dots,n\}$, $x_0\in X_0$ and $p\in P$,
\begin{multline*}
\Phi_i(T;t_0,x_0,p)\in\\
[\Phi_i(T;t_0,\underline{\xi}^i,\underline{\pi}^i)+c^i(x_0-\underline{\xi}^i)+d^i(p-\underline{\pi}^i),\\
\Phi_i(T;t_0,\overline{\xi}^i,\overline{\pi}^i)+c^i(x_0-\overline{\xi}^i)+d^i(p-\overline{\pi}^i)].
\end{multline*}
From (\ref{eq bounded compensation}), $c^i_j\leq0$ (resp.\ $\geq0$) if $s^{x*}_{ij}\geq0$ (resp.\ $\leq0$).
Then for all $x_0\in[\underline{x_0},\overline{x_0}]$, we have $c^i\overline{\xi}^i\leq c^ix_0\leq c^i\underline{\xi}^i$, with $\underline{\xi}^i,\overline{\xi}^i\in X_0$ defined as in (\ref{eq bounded state-param}).
We similarly obtain $d^i\overline{\pi}^i\leq d^ip\leq d^i\underline{\pi}^i$ for all $p\in[\underline{p},\overline{p}]$, which finally leads to the over-approximation in the theorem statement.
\end{proof}

\begin{remark}
\label{rmk tightness}
Unlike the sign-stable case (Lemma~\ref{lemma stable sensi}, Corollary~\ref{coro tightness}), tightness of the over-approximation $\bar R(T;t_0,X_0,P)$ cannot be guaranteed in the general case of Theorem~\ref{th bounded sensi} due to the additional terms $\pm c^i(\underline{\xi}^i-\overline{\xi}^i)$ and $\pm d^i(\underline{\pi}^i-\overline{\pi}^i)$.
% Although the over-approximation $\bar R(T;t_0,X_0,P)$ is tight in the sign-stable case (Lemma~\ref{lemma stable sensi}, Corollary~\ref{coro tightness}), tightness cannot be guaranteed in the more general case of Theorem~\ref{th bounded sensi} due to the additional terms $\pm c^i(\underline{\xi}^i-\overline{\xi}^i)$ and $\pm d^i(\underline{\pi}^i-\overline{\pi}^i)$.
\end{remark}

Following the comparison with discrete-time mixed-monotonicity in Section~\ref{sub stable mixed mono}, a side product of Theorem~\ref{th bounded sensi} is a new over-approximation method for discrete-time systems generalizing the approach from~\cite{coogan2015efficient}.
\begin{coro}
\label{coro mixed mono}
% Let the discrete-time system $x^+=F(x,p)$ have bounded Jacobian matrices $\frac{\partial F(x,p)}{\partial x}\in\mathcal{I}^{n\times n}$ and $\frac{\partial F(x,p)}{\partial p}\in\mathcal{I}^{n\times q}$ over all states $x\in[\underline{x_0},\overline{x_0}]$ and parameters $p\in[\underline{p},\overline{p}]$.
% Then, the reachable set $F(X_0,P)$ after one step can be over-approximated as follows on each dimension $i\in\{1,\dots,n\}$:
% \begin{multline*}
% F_i(X_0,P)\subseteq
% [F_i(\underline{\xi}^i,\underline{\pi}^i)-c^i(\underline{\xi}^i-\overline{\xi}^i)-d^i(\underline{\pi}^i-\overline{\pi}^i),\\
% F_i(\overline{\xi}^i,\overline{\pi}^i)+c^i(\underline{\xi}^i-\overline{\xi}^i)+d^i(\underline{\pi}^i-\overline{\pi}^i)],
% \end{multline*}
% where $\underline{\xi}^i,\overline{\xi}^i,\underline{\pi}^i,\overline{\pi}^i$ and $c^i,d^i$ are defined as in (\ref{eq bounded state-param}) and (\ref{eq bounded compensation}) using the bounds of the Jacobian matrices.
Let $x^+=F(t,x,p)$ have bounded Jacobian matrices $\frac{\partial F(t,x,p)}{\partial x}\in\mathcal{I}^{n\times n}$ and $\frac{\partial F(t,x,p)}{\partial p}\in\mathcal{I}^{n\times q}$ over all states $x\in[\underline{x_0},\overline{x_0}]$ and parameters $p\in[\underline{p},\overline{p}]$.
Then the reachable set $F(t,X_0,P)$ after one step can be over-approximated as follows in each dimension $i\in\{1,\dots,n\}$:
\begin{multline*}
F_i(t,X_0,P)\subseteq
[F_i(t,\underline{\xi}^i,\underline{\pi}^i)-c^i(\underline{\xi}^i-\overline{\xi}^i)-d^i(\underline{\pi}^i-\overline{\pi}^i),\qquad\\
F_i(t,\overline{\xi}^i,\overline{\pi}^i)+c^i(\underline{\xi}^i-\overline{\xi}^i)+d^i(\underline{\pi}^i-\overline{\pi}^i)],
\end{multline*}
where $\underline{\xi}^i,\overline{\xi}^i,\underline{\pi}^i,\overline{\pi}^i$ and $c^i,d^i$ are defined as in (\ref{eq bounded state-param}) and (\ref{eq bounded compensation}) using the bounds of the Jacobian matrices.
\end{coro}

%%%%%%%%%%%%%%%%%%%%%%%%%%%%%%%%%%%%%%%%%%%
%%%%%%%%%%%%%%%%%%%%%%%%%%%%%%%%%%%%%%%%%%%
\section{Obtaining bounds on the sensitivities}
\label{sec bounds}
The approach presented above relies on the trajectory $\Phi(\cdot;t_0,x_0,p):[t_0,+\infty)\rightarrow X$ evaluated at time $T\geq t_0$, which is rarely known explicitly.
Although the successors $\Phi(T;t_0,x_0,p)$ can be computed through numerical integration of the system $\dot x=f(t,x,p)$, the main challenge is the computation of the sensitivity matrices $s^x(T;t_0,x_0,p)$ in (\ref{eq sensitivity state}) and $s^p(T;t_0,x_0,p)$ in (\ref{eq sensitivity parameter}) for all $x_0\in X_0$ and $p\in P$ to evaluate the sign-stability or boundedness of these sensitivities as in Assumptions~\ref{assum stable} and~\ref{assum bounded}, respectively.

%%%%%%%%%%%%%%%%%%%%%%%%%%%%%%%%%%%%%%%%%%%
\subsection{Sampling and falsification}
\label{sub bounds sampling}
In this section, we propose a simulation-based approach where we first evaluate the sensitivity bounds from a few samples in $X_0\times P$ and then use a falsification method to iteratively enlarge these bounds by looking for other pairs in $X_0\times P$ whose sensitivity does not belong to these bounds.

From the definition of $s^x$ in (\ref{eq sensitivity state}), we can use the chain rule to define the time-varying linear system
\begin{equation}
\label{eq sensitivity state system}
\dot s^x(t;t_0,x_0,p)=D_f^x|_\Phi s^x(t;t_0,x_0,p),
\end{equation}
where $D_f^x|_\Phi=D_f^x(t,\Phi(t;t_0,x_0,p),p)$ denotes the Jacobian $D_f^x(t,x,p)=\frac{\partial f(t,x,p)}{\partial x}$ evaluated along the trajectory $\Phi(t;t_0,x_0,p)$.
System (\ref{eq sensitivity state system}) is initialized with the identity matrix $s^x(t_0;t_0,x_0,p)=I_n\in\R^{n\times n}$~\cite{donze2007systematic}.
A similar time-varying affine system can be found for the sensitivity $s^p$:
\begin{equation}
\label{eq sensitivity param system}
\dot s^p(t;t_0,x_0,p)=D_f^x|_\Phi s^p(t;t_0,x_0,p)+D_f^p|_\Phi,
\end{equation}
where $D_f^p|_\Phi$ is the evaluation of $D_f^p(t,x,p)=\frac{\partial f(t,x,p)}{\partial p}$ along the trajectory $\Phi(t;t_0,x_0,p)$ and (\ref{eq sensitivity param system}) is initialized with the zero matrix $s^p(t_0;t_0,x_0,p)=0_{n\times q}\in\R^{n\times q}$~\cite{khalil2001nonlinear}.

For a given time $T\geq t_0$, we first compute the sensitivity matrices $s^x(T;t_0,x_0,p)$ and $s^p(T;t_0,x_0,p)$ through the numerical integration of the systems (\ref{eq sensitivity state system}) and (\ref{eq sensitivity param system}) for at least one pair $(x_0,p)\in X_0\times P$ to obtain initial sensitivity bounds denoted as $[\underline{s^x},\overline{s^x}]\in\mathcal{I}^{n\times n}$ and $[\underline{s^p},\overline{s^p}]\in\mathcal{I}^{n\times q}$.
More than one pair $(x_0,p)$ can be obtained through either random sampling or a gridded discretization of $X_0\times P$.

The second step aims to falsify these bounds~\cite{kapinski2016simulation} through an optimization problem, i.e.\ to find $x_0\in X_0$ and $p\in P$ such that either $s^x(T;t_0,x_0,p)\notin[\underline{s^x},\overline{s^x}]$ or $s^p(T;t_0,x_0,p)\notin[\underline{s^p},\overline{s^p}]$.
Focusing on the sensitivity with respect to the initial state, we want to solve the following optimization problem
\begin{equation*}
% \min_{\substack{x_0\in X_0\\ p\in P}}\left(\min_{i,j}\left(\frac{\overline{s^x_{ij}}-\underline{s^x_{ij}}}{2}-\left|s^x_{ij}(T;t_0,x_0,p)-\frac{\underline{s^x_{ij}}+\overline{s^x_{ij}}}{2}\right|\right)\right)
\min_{\substack{x_0\in X_0\\ p\in P}}\left(\min_{i,j}\left(\frac{\overline{s^x_{ij}}-\underline{s^x_{ij}}}{2}-\left|s^x_{ij}(T;t_0,x_0,p)-s_{ij}^{x*}\right|\right)\right),
\end{equation*}
where for each pair $(i,j)$ we consider a negative absolute value function centered on $s_{ij}^{x*}$ and translated such that the global cost function is negative if and only if there exist $i,j\in\{1,\dots,n\}$ such that $s^x_{ij}(T;t_0,x_0,p)\notin[\underline{s^x_{ij}},\overline{s^x_{ij}}]$.
If the obtained local minimum is negative and the corresponding arguments are denoted as $x_0^*\in X_0$ and $p^*\in P$, the sensitivity bounds are  updated as: $\underline{s^x}\leftarrow\min(\underline{s^x},s^x(T;t_0,x_0^*,p^*))$, $\overline{s^x}\leftarrow\max(\overline{s^x},s^x(T;t_0,x_0^*,p^*))$, using elementwise $\min$ and $\max$ operators.
This process is repeated with the new bounds until a positive minimum is obtained.
A similar approach is applied to $[\underline{s^p},\overline{s^p}]$.

\begin{remark}
\label{rmk sampling}
While this approach is likely to result in an accurate approximation of the actual sensitivity bounds, it is not guaranteed to over-approximate the set of all possible sensitivity values over $X_0\times P$ since the falsification relies on an optimization problem only able to provide local minima.
\end{remark}

%%%%%%%%%%%%%%%%%%%%%%%%%%%%%%%%%%%%%%%%%%%
\subsection{Interval arithmetics}
\label{sub bounds interval}
An alternative approach recommended in~\cite{xue2017cdc} is based on the use of interval arithmetics to solve an affine time-varying system as presented in~\cite{althoff2007reachability}.
For the purpose of comparison with the method in Section~\ref{sub bounds sampling} on the numerical examples of Section~\ref{sec simu}, we give an overview of how the results described in~\cite{althoff2007reachability} can be applied to the sensitivity systems (\ref{eq sensitivity state system}) and (\ref{eq sensitivity param system}) to obtain guaranteed bounds on the sensitivity matrices.
We start from the assumption that bounds on the Jacobian matrices $D_f^x(t,x,p)=\frac{\partial f(t,x,p)}{\partial x}$ and $D_f^p(t,x,p)=\frac{\partial f(t,x,p)}{\partial p}$ of (\ref{eq system}) are known or can be computed.
\begin{assum}
\label{assum jacobian}
Given an invariant set $X\subseteq\R^n$ of (\ref{eq system}), there exist interval matrices $\mathcal{A}\in\mathcal{I}^{n\times n}$ and $\mathcal{B}\in\mathcal{I}^{n\times q}$ such that for all $t\in[t_0,+\infty)$, $x\in X$, $p\in P$, we have $D_f^x(t,x,p)\in\mathcal{A}$ and $D_f^p(t,x,p)\in\mathcal{B}$.
\end{assum}

We can then rewrite (\ref{eq sensitivity state system}) and (\ref{eq sensitivity param system}) as the set-valued systems
\begin{align}
\label{eq sensitivity state matrix}
&\dot s^x(t)\in\mathcal{A} s^x(t), &s^x(t_0)=I_n,\\
\label{eq sensitivity param matrix}
&\dot s^p(t)\in\mathcal{A} s^p(t)+\mathcal{B}, &s^p(t_0)=0_{n\times q}.
\end{align}
The solution of these systems at time $T$ is over-approximated using interval arithmetics and a truncated Taylor series of the interval matrix exponential $e^{\mathcal{A}(T-t_0)}\in\mathcal{I}^{n\times n}$, detailed in~\cite{althoff2007reachability}.
\begin{lemma}[\cite{althoff2007reachability}]
\label{lemma sensitivity bounds}
Under Assumption~\ref{assum jacobian}, there exist functions $m:[0,+\infty)\rightarrow\N$ and $E:[0,+\infty)\rightarrow\mathcal{I}^{n\times n}$ defined in~\cite{althoff2007reachability} such that for all $T\geq t_0$, we have 
\begin{gather*}
s^x(T)\in\sum_{i=0}^{m(T-t_0)}\frac{(\mathcal{A}(T-t_0))^i}{i!}+E(T-t_0),\\
s^p(T)\in\left(\sum_{i=0}^{m(T-t_0)}\frac{(\mathcal{A}(T-t_0))^i}{(i+1)!}+E(T-t_0)\right)(T-t_0)\mathcal{B}.
\end{gather*}
\end{lemma}

\begin{remark}
\label{rmk interval arithmetics}
Unlike the sampling-based method in Section~\ref{sub bounds sampling}, Lemma~\ref{lemma sensitivity bounds} provides guaranteed over-approximations for the sensitivities but risks being overly conservative since interval arithmetics cannot provide exact set computation when more than two interval matrices are multiplied~\cite{jaulin2001applied}.
\end{remark}

\begin{remark}
\label{rmk taylor}
The minimal Taylor order $m(T-t_0)$ for Lemma~\ref{lemma sensitivity bounds} to hold is a linearly increasing function of the time step $T-t_0$~\cite{althoff2007reachability}.
This approach might thus be practically infeasible when the desired time step $T-t_0$ is too large.
\end{remark}

%%%%%%%%%%%%%%%%%%%%%%%%%%%%%%%%%%%%%%%%%%%
%%%%%%%%%%%%%%%%%%%%%%%%%%%%%%%%%%%%%%%%%%%
\section{Numerical examples}
\label{sec simu}
All computations are run with Matlab on a laptop with a 1.7GHz CPU and 4GB of RAM.

%%%%%%%%%%%%%%%%%%%%%%%%%%%%%%%%%%%%%%%%%%%
\subsection{Traffic network}
\label{sub simu traffic}
Consider the $3$-link traffic network describing a \emph{diverge} junction (the vehicles in link $1$ divide evenly among the outgoing links $2$ and $3$) inspired by~\cite{coogan2016benchmark}:
\begin{equation}
\label{eq traffic CT}
\dot x = \frac{1}{T}
\begin{pmatrix}
p-g(x)\\
g(x)/2-\min(c,vx_2)\\
g(x)/2-\min(c,vx_3)
\end{pmatrix},
\end{equation}
where $g(x)=\min(c,vx_1,2w(\bar x-x_2),2w(\bar x-x_3))$, $x\in\R^3$ is the vehicle density on the three links, $p\in P=[40,60]$ is the constant but uncertain vehicle inflow to link $1$, $T=30$ seconds and $c=40$, $v=0.5$, $\bar x=320$, $w=1/6$ are known parameters of the network detailed in~\cite{coogan2016benchmark}.

In addition to the uncertain disturbance input $p\in P=[40,60]$, we consider a set $X_0\subseteq\mathcal{I}^3$ of initial conditions such that $X_0=[150,200]\times[250,320]\times[50,100]$, meaning that link $2$ is close to its maximal capacity $\bar x=320$ while link $3$ has more availability.
Figure~\ref{fig traffic} presents the projection in the $(x_1,x_3)$ plane of the initial interval $X_0$ (dashed black), the reachable set $\{\Phi(T;x_0,p)~|~x_0\in X_0,p\in P\}$ (hatched black) of (\ref{eq traffic CT}) after time $T=30$ seconds and two interval over-approximations of this set obtained as described below.

Using the simulation-based approach in Section~\ref{sub bounds sampling}, we get a first approximation of the sensitivity bounds of (\ref{eq traffic CT}) from a grid of $16$ samples in $X_0\times P$ ($2$ samples per dimension) computed in $0.98$s, which is then refined in $15.8$s through falsification, stopping after $6$ iterations.
From these times, it is thus advised to use a finer sampling of $X_0\times P$ to obtain a good initial estimation of the sensitivity bounds so that the number of falsification runs is reduced.
%As expected from the sign-stability of the Jacobian matrices for the discrete-time counterpart of (\ref{eq traffic CT}) described in~\cite{coogan2016benchmark}, 
The numerical computations indicate that the sensitivity bounds for (\ref{eq traffic CT}) are sign-stable, thus leading to the tight red over-approximation in Figure~\ref{fig traffic} obtained after applying Lemma~\ref{lemma stable sensi}.

The second over-approximation in green is computed from sensitivity bounds obtained with the interval arithmetics approach in Lemma~\ref{lemma sensitivity bounds}, where the Jacobian bounds as in Assumption~\ref{assum jacobian} are obtained analytically from the dynamics (\ref{eq traffic CT}).
We pick a Taylor order $m=7$ (empirically, we see no improvement on the sensitivity bounds for larger values) which is greater than the minimal value $m(T)=0$ for Lemma~\ref{lemma sensitivity bounds} to hold.
The sensitivity bounds are computed in $14$ms, but as predicted in Remark~\ref{rmk interval arithmetics} they are much more conservative than the one obtained in the first approach and they do not satisfy Assumption~\ref{assum stable}.
The over-approximation in green is thus obtained from the generalized result in Theorem~\ref{th bounded sensi} and is much larger than the red one, firstly because Theorem~\ref{th bounded sensi} is known not to be tight (Remark~\ref{rmk tightness}), but also because it tries to compensate for the sensitivity elements believed not to be sign-stable while their real values are actually sign-stable according to the sampling-based estimation above.

The computation of both red and green over-approximations (from Lemma~\ref{lemma stable sensi} and Theorem~\ref{th bounded sensi}) is done in $60$ms.
The volumes of the red and green over-approximations are respectively $1.7$ and $5.1$ times the volume of the true reachable set.

\begin{figure}[tbh]
\centering
\includegraphics[width=\columnwidth]{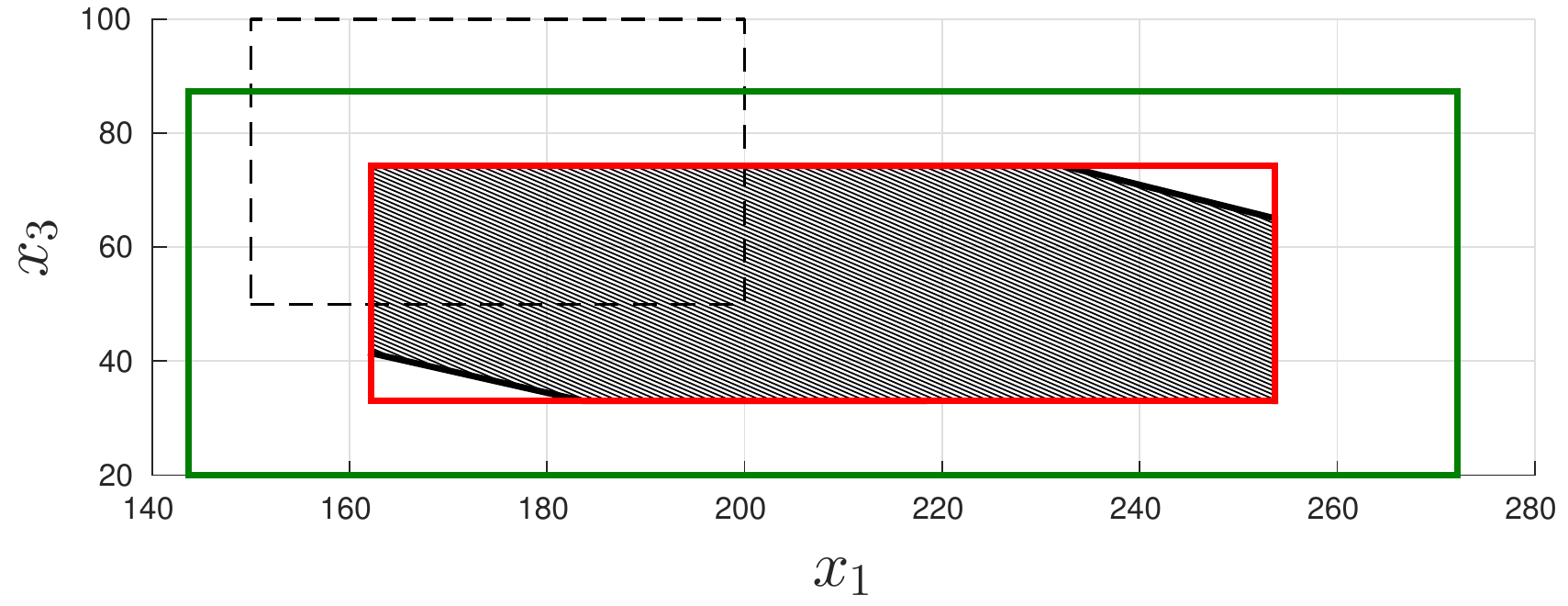}
\caption{$(x_1,x_3)$-projection of the reachable set (hatched black) of (\ref{eq traffic CT}) from the initial interval (dashed black) and its over-approximations using sampling-based sensitivity bounds (red) and interval arithmetics (green).}
\label{fig traffic}
\end{figure}

To study the scalability of the approach, we now extend this three link example by adding links downstream of the diverging junction so that traffic on link $2$ flows to link $4$ then to link $6$, \emph{etc.}, and, likewise, traffic flows from link $3$ to $5$ to $7$, \emph{etc.}
The modified dynamics are
\begin{align}
   \label{eq traffic large1}
  f_{i}(x,p)&=\frac{1}{T}(g(x)/2-h(x_i,x_{i+2})),\ i\in\{2,3\}\\
  \label{eq traffic large2}
  f_{i}(x,p)&=\frac{1}{T}(\beta h(x_{i-2},x_i)-h(x_i,x_{i+2})),\ i\in\{4,\ldots,n\}
\end{align}
where $  h(\eta,\zeta)=\min\left(c, v\eta,\frac{w}{\beta}(\bar{x}-\zeta)\right)$, 
$n$ is the total number of links in the network, and we take $\beta=\frac{3}{4}$ ($1-\beta$ is the fraction of vehicles exiting the network after each link).
For $i\in\{n-1,n\}$, the term $\frac{w}{\beta}(\bar{x}-x_{i+2})$ is excluded from the minimization in $h$.
Considering a $11$-link network with $X_0=[20,300]^{11}$, we apply the same methods as for the previous $3$-link case.
The sensitivity bounds are first evaluated from a grid of $4096$ samples of $X_0\times P$ ($2$ samples per dimension) in $862$s, followed by $3$ iterations of falsification in $21$s, resulting in sign-stable bounds.
Another set of bounds is computed in $0.54$s through interval arithmetics with a Taylor order $m=15$, resulting in bounds which are not sign-stable.
The over-approximations in the state space (from Lemma~\ref{lemma stable sensi} and Theorem~\ref{th bounded sensi}) using both sets of sensitivity bounds are computed in $0.25$s.
From the sign-stability assumption, the first interval over-approximation is guaranteed to be tight to the actual reachable set (Corollary~\ref{coro tightness}).
On the other hand, the over-approximation obtained from interval arithmetics is not tight and has $68$ times the volume of the first over-approximation, making it too loose for practical use.
From the computation times for both the $3$-link and $11$-link models, we note that the approach scales well with the state dimension apart from the main bottleneck in the sampling approach, whose complexity grows exponentially with $n$ for a gridded sampling.

%%%%%%%%%%%%%%%%%%%%%%%%%%%%%%%%%%%%%%%%%%%
\subsection{Satellite orbit}
\label{sub simu satellite}
Consider the non-linear system describing a satellite orbiting a celestial body from~\cite{thomson2012introduction}:
\begin{equation}
\label{eq satellite}
\dot x =
\begin{pmatrix}
x_2\\
-\frac{p}{x_1^2}+x_1x_4^2\\
x_4\\
-\frac{2x_2x_4}{x_1}
\end{pmatrix},\quad
x(0)=
\begin{pmatrix}
R+400\\
0\\
0\\
\sqrt{\frac{p}{(R+400)^3}}
\end{pmatrix},
\end{equation}
where $x_1$ is the distance of the satellite to the center of the body, $x_3$ its angular position and $x_2$ and $x_4$ their respective derivatives.
The parameter $p\in\R$ is defined as $p=GM$, where $G$ is the gravitational constant and $M$ the mass of the body.
The initial conditions of (\ref{eq satellite}) are chosen to obtain a circular orbit at $400$km above the body's surface (radius $R$).
Assuming uncertain values (around Earth's known values) for both the parameter $p\in[3.9779,3.9938]\cdot 10^5$ km$^3/$s$^2$ and the desired orbit radius $R+400\in[6.7718,6.7845]\cdot 10^3$ km, we obtain uncertainty bounds denoted as $p\in P\subseteq\R$ and $x(0)\in X_0\in\mathcal{I}\times\{0\}\times\{0\}\times\mathcal{I}\subseteq\R^4$.

We want to study the effect of these uncertainties on the reachable set of (\ref{eq satellite}) at time $T=92$ minutes (after approximately one whole revolution around the Earth).
As expected from Remark~\ref{rmk taylor}, the interval arithmetics result from Lemma~\ref{lemma sensitivity bounds} is not applicable to (\ref{eq satellite}) since the choice of $T=5520$s requires a minimum Taylor order $m(T)=88883$, which cannot be computed in reasonable time.
We thus rely on the sampling-based approach from Section~\ref{sub bounds sampling} by first evaluating the sensitivity bounds for $100$ random samples in $X_0\times P$, obtained in $68$s.
A single iteration of falsification is then run in $5$s, meaning that the sampling-based approximation of the bounds already covered all sensitivity values that could be found from the optimization problem solved in the falsification.
The obtained sensitivity bounds $[\underline{s^x},\overline{s^x}]\in\mathcal{I}^{4\times 4}$ and $[\underline{s^p},\overline{s^p}]\in\mathcal{I}^4$ do not satisfy the sign-stability condition in Assumption~\ref{assum stable} on $9$ of their $20$ entries, thus requiring the application of the generalized result in Theorem~\ref{th bounded sensi} to compute (in $58$ms) the over-approximation of the reachable set of (\ref{eq satellite}) at time $T$, projected into the polar coordinate system $(x_1,x_3)$ in Figure~\ref{fig satellite} (in blue) along with an estimation of the actual reachable set (cloud of black dots) obtained from $10000$ random samples in $X_0\times P$.
Despite the lack of guarantee in the sampling-based approach (Remark~\ref{rmk sampling}), Figure~\ref{fig satellite} suggests that the computed interval does indeed over-approximate the reachable set and is not overly conservative.

\begin{figure}[tbh]
\centering
\includegraphics[width=\columnwidth]{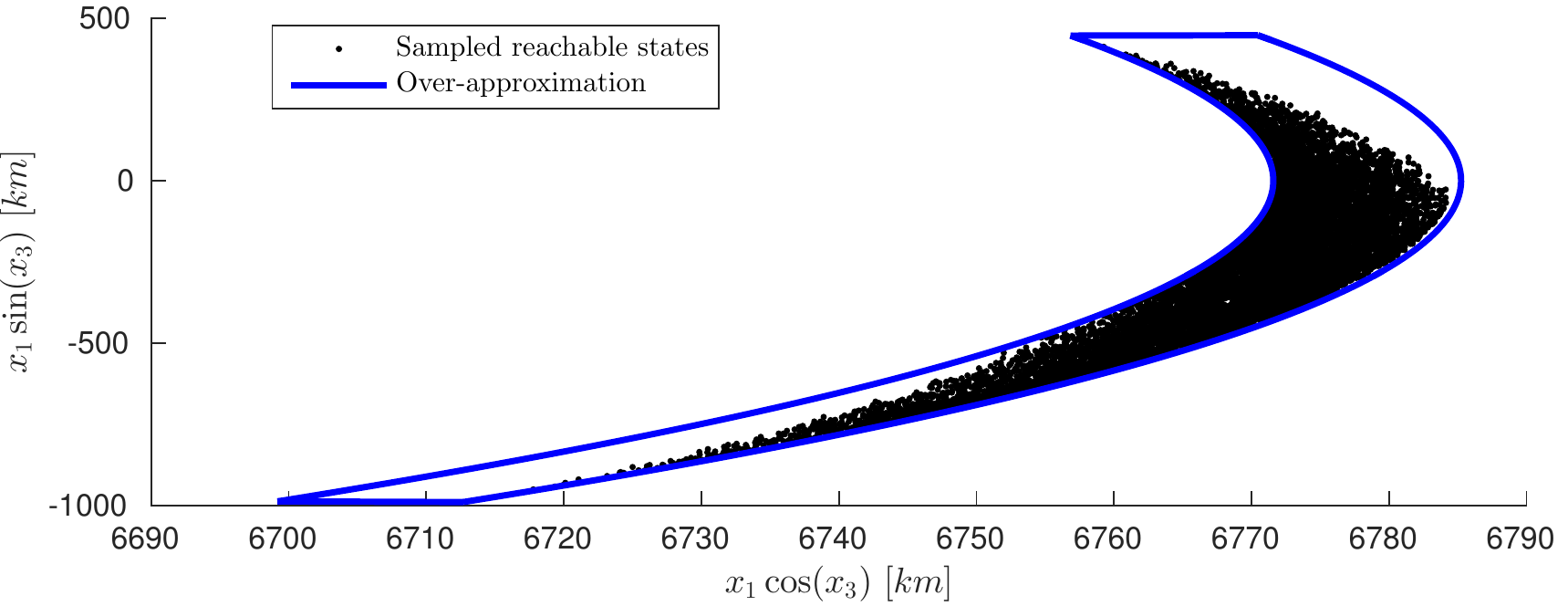}
\caption{Reachable set (black) of (\ref{eq satellite}) and its over-approximation $\bar R(T;0,X_0,P)$ (blue) projected in the polar plane $(x_1,x_3)$.}
\label{fig satellite}
\end{figure}

%%%%%%%%%%%%%%%%%%%%%%%%%%%%%%%%%%%%%%%%%%%
%%%%%%%%%%%%%%%%%%%%%%%%%%%%%%%%%%%%%%%%%%%
\section{Conclusion}
\label{sec conclu}
This paper provides a new reachability analysis method based on the sensitivity matrices of a continuous-time system and applicable to the wide class of systems whose sensitivity matrices at a given time are bounded over the sets of uncertain parameters and initial conditions.
This assumption is very mild since it is naturally satisfied by any system with a sufficiently smooth trajectory function.
The computation of an interval over-approximation of the reachable set using this approach has favorable scalability, since its complexity is at worst linear in the state dimension.

Since the system trajectories or sensitivity matrices are rarely known explicitly, the main challenge of this method lies in obtaining bounds on the sensitivity.
Two such approaches are considered in this paper.
The first approach relies on interval arithmetics and provides guaranteed sensitivity bounds but can rarely be applied in practice, as the bounds are often overly conservative and the computation is infeasible for larger time steps.
The second approach is based on sampling and falsification and provides more reliable values for the sensitivity bounds although without formal guarantees, which may present a risk for safety-critical applications.
The sampling-based approach is currently the main computational bottleneck, since the suggested number of samples to obtain a good first estimate of the sensitivity bounds (in order to minimize the number of falsification iterations) grows exponentially with the state dimension.

Future work will aim to exploit these results for abstraction-based synthesis (see e.g.~\cite{coogan2015efficient}), where a control problem on a differential equation is instead solved on a finite transition system abstracting the continuous dynamics.
In such approaches, reachability analysis plays a central role in the creation of the abstraction and intervals are commonly used for their implementation benefits (low memory requirement, easy to check intersection with other intervals).

%%%%%%%%%%%%%%%%%%%%%%%%%%%%%%%%%%%%%%%%%%%
%%%%%%%%%%%%%%%%%%%%%%%%%%%%%%%%%%%%%%%%%%%
\bibliographystyle{abbrv}
\bibliography{2018_Meyer_LCSS18} 
\end{document}